\documentclass[12pt,reqno,a4paper]{amsart}
\usepackage{amsmath,amssymb,dsfont,graphicx,cite,latexsym,
epsf,cancel,epic,eepic}
\usepackage[colorlinks=false]{hyperref}
\usepackage{mathpazo}
\newtheorem{theorem}{Theorem}

\newtheorem{lemma}{Lemma}

\newtheorem{definition}{Definition}
\newtheorem{corollary}{Corollary}
\def\beq{\begin{equation}}
\def\eeq{\end{equation}}
\def\bea{\begin{eqnarray}}
\def\eea{\end{eqnarray}}
\makeatletter
\expandafter\let\expandafter
\reset@font\csname reset@font\endcsname
\def\subeqnarray{\arraycolsep1pt
   \def\@eqnnum\stepcounter##1{\stepcounter{subequation}
       {\reset@font\rm(\theequation\alph{subequation})}}
\jot5mm     \eqnarray}

\makeatother

\newcommand{\bbR}{{\mathbb R}}
\newcommand{\bbC}{{\mathbb C}}
\newcommand{\bbP}{{\mathbb P}}

\newcommand{\cE}{{\mathcal E}}
\def\ep{\varepsilon}
\def\epsilon{\varepsilon}
\def\t{\widetilde}

\def\nn{\nonumber}

\newbox\meibox
\def\placeunder#1#2#3#4{\setbox\meibox%
\vbox{\hbox{\hskip#4$\hphantom{#2}$}\hbox{$\hphantom{#1}$}}%
\vtop{\baselineskip=0pt\lineskiplimit=\baselineskip%
\lineskip=#3\hbox to \wd\meibox{\hfil\hskip#4$#2$\hfil}%
\hbox to \wd\meibox{\hfil$#1$\hfil}}}
\def\intprod{\mathbin{\hbox to 6pt{%
                 \vrule height0.4pt width5pt depth0pt
                 \kern-.4pt
                 \vrule height6pt width0.4pt depth0pt\hss}}}

\begin{document}
\title[Geometry of Kahan discretizations]
{Geometry of the Kahan discretizations of planar quadratic Hamiltonian systems}

%
\author{Matteo Petrera, Jennifer Smirin \and Yuri B. Suris }

\thanks{E-mail: {\tt  petrera@math.tu-berlin.de, jennifer.smirin@gmail.com, suris@math.tu-berlin.de}}

\maketitle

\begin{center}
{\footnotesize{
Institut f\"ur Mathematik, MA 7-1\\
Technische Universit\"at Berlin, Str. des 17. Juni 136,
10623 Berlin, Germany
}}
\end{center}

\begin{abstract}
Kahan discretization is applicable to any quadratic vector field and produces a birational map which approximates the shift along the phase flow. For a planar quadratic Hamiltonian vector field, this map is known to be integrable and to preserve a pencil of cubic curves. Generically, the nine base points of this pencil include three points at infinity (corresponding to the asymptotic directions of cubic curves) and six finite points lying on a conic. We show that the Kahan discretization map can be represented in six different ways as a composition of two Manin involutions, corresponding to an infinite base point and to a finite base point. As a consequence, the finite base points can be ordered so that the resulting hexagon has three pairs of parallel sides which pass through the three base points at infinity. Moreover, this geometric condition on the base points turns out to be characteristic: if it is satisfied, then the cubic curves of the corresponding pencil are invariant under the Kahan discretization of a planar quadratic Hamiltonian vector field.
\end{abstract}

\section{Introduction}

The Kahan discretization was introduced in the unpublished notes \cite{K} as a method applicable to any system of ordinary differential equations on $\bbR^n$ with a quadratic vector field:
\begin{equation}\label{eq: diff eq gen}
\dot{x}=f(x)=Q(x)+Bx+c,
\end{equation}
where each component of $Q:\bbR^n\to\bbR^n$ is a quadratic form, while $B\in{\rm Mat}_{n\times n}(\bbR)$ and $c\in\bbR^n$. Kahan's discretization reads as
\begin{equation}\label{eq: Kahan gen}
\frac{\widetilde{x}-x}{\epsilon}=Q(x,\widetilde{x})+\frac{1}{2}B(x+\widetilde{x})+c,
\end{equation}
where
\[
Q(x,\widetilde{x})=\frac{1}{2}\big(Q(x+\widetilde{x})-Q(x)-Q(\widetilde{x})\big)
\]
is the symmetric bilinear form corresponding to the quadratic form $Q$. Equation (\ref{eq: Kahan gen}) is {\em linear} with respect to $\widetilde x$ and therefore defines a {\em rational} map $\widetilde{x}=\Phi_f(x,\epsilon)$. Explicitly, one has
\beq \label{eq: Phi gen}
\t x =\Phi_f(x,\ep)= x + \ep \left( I - \frac{\ep}{2} f'(x) \right)^{-1} f(x),
\eeq
where $f'(x)$ denotes the Jacobi matrix of $f(x)$. 

Clearly, this map approximates the time $\epsilon$ shift along the solutions of the original differential system. Since equation (\ref{eq: Kahan gen}) remains invariant under the interchange $x\leftrightarrow\widetilde{x}$ with the simultaneous sign inversion $\epsilon\mapsto-\epsilon$, one has the {\em reversibility} property
\begin{equation}\label{eq: reversible}
\Phi^{-1}_f(x,\epsilon)=\Phi_f(x,-\epsilon).
\end{equation}
In particular, the map $f$ is {\em birational}. 
Kahan applied this discretization scheme to the famous Lotka-Volterra system and showed that in this case it possesses a very remarkable non-spiralling property. This property was explained by Sanz-Serna \cite{SS} by demonstrating that in this case the numerical method preserves an invariant Poisson structure of the original system.

The next intriguing appearance of this discretization was in the two papers by Hirota and Kimura who (being apparently unaware of the work by Kahan) applied it to two famous {\em integrable} system of classical mechanics, the Euler top and the Lagrange top \cite{HK, KH}. Surprisingly, the discretization scheme produced in both cases {\em integrable} maps. 

In \cite{PS, PPS1, PPS2} the authors undertook an extensive study of the properties of the Kahan's method when applied to integrable systems (we proposed to use in the integrable context the term ``Hirota-Kimura method''). It was demonstrated that, in an amazing number of cases, the method preserves integrability in the sense that the map $\Phi_f(x,\epsilon)$ possesses as many independent integrals of motion as the original system $\dot x=f(x)$.

Further remarkable geometric properties of the Kahan's method were discovered by Celledoni, McLachlan, Owren and Quispel.
\begin{theorem} {\bf\cite{CMOQ1}}
Consider a Hamiltonian vector field $f(x)=J\nabla H(x)$, where $J$ is a skew-symmetric $n\times n$ matrix, and the Hamilton function $H:\bbR^n \to\bbR$ is a polynomial of degree 3. Then the map $\Phi_f(x,\epsilon)$ possesses the following rational integral of motion: 
\beq\label{eq: CMOQ integral}
\t H(x,\epsilon) = H(x)+\frac{\ep}{3}  (\nabla H(x))^{\rm T} \left( I - \frac{\ep}{2} f'(x) \right)^{-1} f(x),
\eeq
as well as an invariant measure
\beq\label{eq: CMOQ measure}
\frac{dx_1\wedge\ldots\wedge dx_n}{\displaystyle{\det\left( I - \frac{\ep}{2} f'(x) \right)}}.
\eeq
The degree of the denominator $\det( I - \frac{\ep}{2} f'(x))$ of the function $\t H(x,\epsilon)$ is $n$, while the degree of the numerator of $\t H(x,\epsilon)$ is $n+1$.
\end{theorem}

In the present paper, we will be studying the case $n=2$. Here, $\Phi_f$ is a birational planar map with an invariant measure and an integral of motion, thus completely integrable. Integral (\ref{eq: CMOQ integral}) in this case is given by 
\beq\label{eq: Kahan int}
\t H(x,y,\ep)= \frac{C(x,y,\ep)}{D(x,y,\ep)},
\eeq
where $C(x,y,\ep)$ is a polynomial in $x,y$ of degree 3, and $D(x,y,\ep)$ is a polynomial of degree 2. Thus, the level sets of the integral are cubic curves
$$
\cE_\lambda=\big\{(x,y): C(x,y,\ep)-\lambda D(x,y,\ep)=0\big\},
$$ 
which form a linear system (a pencil). Such a pencil is characterized by its base points (common points of all curves of the pencil). On each invariant curve, the map $\Phi_f$ induces a shift (respective to the addition law on this curve). In \cite{CMOQ2, KCMMOQ} it was shown that in several cases $\Phi_f$ can be represented as a {\em Manin transformation}, i.e., as a composition of two {\em Manin involutions}, as defined, e.g., in \cite[p. 35]{Ves}, \cite[Sect. 4.2]{Dui}.

In the present paper we prove that, actually, in the generic situation $\Phi_f$ admits six different representations as a Manin transformation.  Analyzing these representations, we arrive at an amazing geometric characterization:
\begin{itemize}
\item[] {\em A pencil of elliptic curves consists of invariant curves for the Kahan discretization of a planar quadratic Hamiltonian vector field if and only if  the six finite base points can be ordered so that the resulting hexagon has three pairs of parallel sides which pass through the three base points at infinity.}
\end{itemize}

The structure of the paper is as follows. In Section \ref{sect: pencils and manin}, we discuss the generalities about the pencils of cubic curves for Kahan discretizations, recall the difinitions of Manin involutions and Manin maps, as well as formulas for computing those maps. In Section \ref{sect: 6-fold}, we prove the six-fold representation of the Kahan discretization of a generic planar quadratic Hamiltonian vector field as a Manin map. Finally, in Section \ref{sect: geometry} we derive from the latter fact the remarkable geometric characterization given above.

\section{The geometry of the invariant cubic curves of the Kahan map}
\label{sect: pencils and manin}

The geometric properties of the birational map $\Phi_f$ become more uniform if we consider it in the complex domain,  i.e., as a map of the complex plane $\bbC^2$, and actually as a map of the projective complex plane $\bbC\bbP^2$. In particular, we assume that all the coefficients $a_i$ of the Hamilton function $H(x,y)$ are complex numbers, as well as its arguments $x,y$. 

This phase space is foliated by the one-parameter family (pencil) of invariant curves 
$$
\mathcal{E}_\lambda= \left\{ (x,y) \in \bbC^2 \,:\,
C(x,y,\ep) - \lambda D(x,y,\ep) =0\right\}.
$$
Indeed, for any initial point $(x_0,y_0)\in \bbC^2$ the orbit of $\Phi$  lies on the cubic curve $\cE_\lambda$ with $\lambda=\t H(x_0,y_0,\ep)=C(x_0,y_0,\ep)/D(x_0,y_0,\ep)$.

We consider $\bbC^2$ as an affine part of $\bbC\bbP^2$ consisting of the points $[x:y:z]\in\bbC\bbP^2$ with $z\neq 0$.
We define the projective curves $\bar\cE_\lambda$ as projective completion of $\cE_\lambda$: 
$$
\bar {\mathcal{E}}_\lambda= 
\left\{ [x:y:z] \in \mathbb{C}\mathbb{P}^2\,:\,
\bar C(x,y,z,\ep) - \lambda z\bar D(x,y,z,\ep) =0
\right\},
$$
where we set
$$
\bar C(x,y,z,\ep)=z^3C(x/z,y/z,\ep), \quad \bar D(x,y,z,\ep)=z^2D(x/z,y/z,\ep).
$$
We assume that the curve $\bar {\mathcal{E}}_0= \left\{ [x:y:z] \in \mathbb{C}\mathbb{P}^2\,:\,\bar C(x,y,z,\ep)  =0\right\}$ is nonsingular. Note that the second basis curve of the pencil, $\bar {\mathcal{E}}_\infty= \left\{ [x:y:z] \in \mathbb{C}\mathbb{P}^2\,:\,z\bar D(x,y,z,\ep)  =0\right\}$, is reducible, and consists of the conic $\{\bar D(x,y,z,\ep)=0\}$ and of the line at infinity $\{z=0\}$.

All curves $\bar\cE_\lambda$ pass through the set of {\em base points} which is defined as $\bar\cE_0\cap\bar\cE_\infty$. According to the Bezout theorem, there are nine base points, counted with multiplicities. In our specific context, there are three {\em base points at infinity}:
$$
\{F_1,F_2,F_3\}=\bar\cE_0\cap \{z=0\}, 
$$
and six further base points $\{B_1,\ldots B_6\}=\bar\cE_0\cap\{\bar D=0\}$. See Figure \ref{Fig1}. For any point of $P\in\bbC\bbP^2$ different from the base points, there is a unique  curve $\bar\cE_\lambda$ of the pencil such that $P\in\bar\cE_\lambda$.

\begin{figure}
\begin{center}
\includegraphics[width=0.5\textwidth]{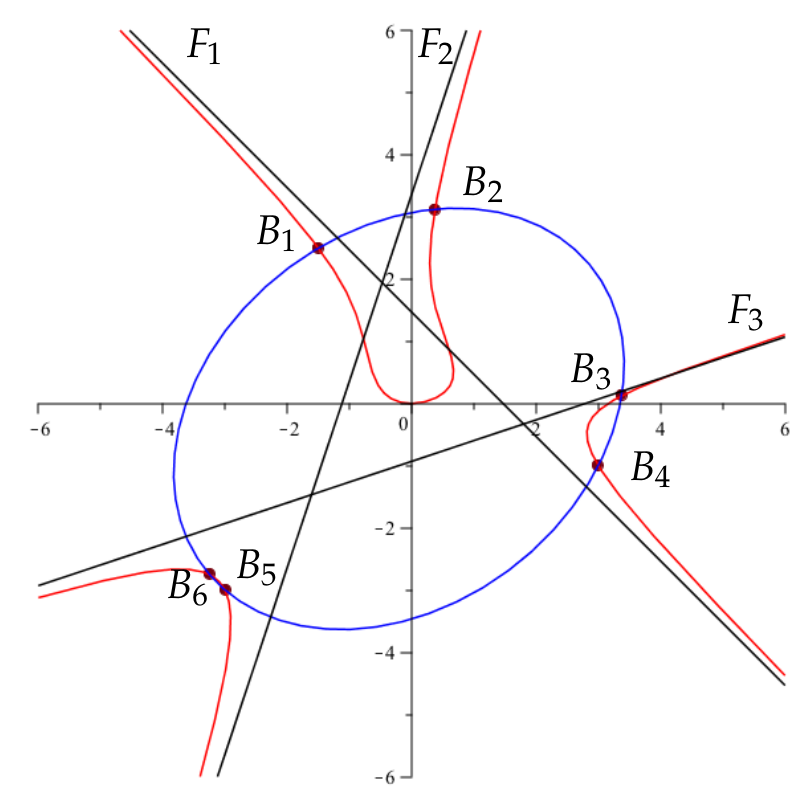}
\caption{The red curve is the cubic $C=0$ passing through the origin, the blue curve is the conic $D=0$. The black lines are the tangent lines to the red cubic at the points $F_1$, $F_2$, $F_3$ at infinity.  The finite base points are $B_1, \ldots, B_6$.}
\label{Fig1}
\end{center}
\end{figure}

Now we recall the definitions of Manin involutions and Manin transformations, following \cite[Sect. 4.2]{Dui}.  
\begin{definition}
\quad

{\em 1)}
Consider a nonsingular cubic curve $\bar\cE$ in $\bbC\bbP^2$, and a point $P_0\in\bar\cE$. The {\em Manin involution} on $\bar{\mathcal{E}}$ with respect to $P_0$ is the map $I_{\bar\cE, P_0}: \bar{\mathcal{E}}\rightarrow \bar{\mathcal{E}}$ defined as follows:
\begin{itemize}
\item For $P_1\neq P_0$, the point $P_2=I_{\bar\cE, P_0}(P_1)$ is the unique third intersection point of $\bar\cE$ with the line $(P_0P_1)$;
\item For $P_1=P_0$, the point $P_2=I_{\bar\cE, P_0}(P_1)$ is the unique second intersection point of $\bar\cE$ with the tangent line to $\bar\cE$ at $P_0=P_1$.
\end{itemize} 

{\em 2)} Let $P_0,P_1$ be two distinct points of the curve $\bar\cE$. The {\em Manin transformation} $M_{\bar\cE,P_0,P_1}:\bar\cE\to\bar\cE$ is defined as
\beq\label{Manin map loc}
M_{\bar\cE,P_0,P_1}=I_{\bar\cE,P_1}\circ I_{\bar\cE,P_0}.
\eeq
\end{definition}
It is easy to see that $M_{\bar\cE,P_0,P_1}(P_1)=P_0$. Indeed, if $P_2$ is the third intersection point of $\bar\cE$ with the line $(P_0P_1)$ then $I_{\bar\cE,P_0}(P_1)=P_2$ and $I_{\bar\cE,P_1}(P_2)=P_0$. With a natural addition law on the elliptic curve $\bar\cE$ (with a flex $O$ as a neutral element, so that $P_0+P_1+P_2=O$ if and only if the points $P_0$, $P_1$, $P_2$ are collinear), we can write $I_{\bar\cE,P_0}(P)=-(P_0+P)$, and
$$
M_{\bar\cE,P_0,P_1}(P)=I_{\bar\cE,P_1}(I_{\bar\cE,P_0}(P))=-(P_1+I_{\bar\cE,P_0}(P))=-(P_1-(P_0+P))=P+P_0-P_1,
$$  
so that $M_{\bar\cE,P_0,P_1}$ is the translation by $P_0-P_1$ on the elliptic curve $\bar\cE$. 
\begin{definition}
Consider a pencil $\frak E=\{\bar\cE_\lambda\}$ of cubic curves in $\bbC\bbP^2$ with at least one nonsingular member.

{\em 1)}   Let $B$ be a base point of the pencil. The {\em Manin involution} $I_{\frak E,B}:\bbC\bbP^2\dashrightarrow\bbC\bbP^2$ is a birational map defined as follows. For any $P\in\bbC\bbP^2$ which is not a base point, $I_{\frak E,B}(P)=I_{\bar\cE_\lambda,B}(P)$, where $\bar\cE_\lambda$ is the unique curve of the pencil containing the point $P$. 

{\em 2)}  Let $B_1,B_2$ be two distinct base points of the pencil. The {\em Manin transformation} $M_{\frak E,B_1,B_2}:\bbC\bbP^2\dashrightarrow\bbC\bbP^2$ is a birational map defined as
\beq\label{Manin map}
M_{\frak E,B_1,B_2}=I_{\frak E,B_2}\circ I_{\frak E,B_1}.
\eeq
\end{definition}

 We now provide formulas which can be used to compute Manin involutions.

\begin{lemma} \label{lemma involution inf}
Let $\bar\cE$ be a nonsingular cubic curve in $\bbC\bbP^2$ given in the non-homogeneous coordinates by the equation
$$
\cE: \quad u_1x^3+u_2x^2y+u_3xy^2+u_4y^3+u_5x^2+u_6xy+u_7y^2+u_8x+u_9y+u_{10}=0.
$$
Let  $P_0=(x_0,y_0)$ be a finite point on ${\mathcal{E}}$, and let $F=[1: \mu : 0]$ be a point of $\bar{\mathcal{E}}$ at infinity. If $P_1=I_{\bar\cE,P_0}(F)$ is finite, then it is given by
$$
P_1 = \Big( - (x_0 +\alpha), - \mu (x_0  + \alpha) + \beta \Big),
$$
where
\bea
&& \beta= y_0 - \mu x_0, \nn \\
&& \alpha= \frac{(3u_4\mu+u_3)\beta^2+(2u_7\mu+u_6)\beta+u_9\mu+u_8}{(3 u_4\mu^2 +2u_3\mu+u_2)\beta +u_7 \mu^2 + u_6 \mu +u_5}. \nn 
\eea
\end{lemma}
\begin{proof} Equation of the line $(P_0F)$ is $y=\mu x+\beta$. To find the second (finite) intersection point $P_1=(x_1,y_1)$ of this line with $\cE$, we substitute this equation into equation $U(x,y)=0$ of the cubic curve. We have:
\beq \label{involution cubic eq}
U(x,\mu x+\beta)=A_3x^3+A_2x^2+A_1x+A_0,
\eeq
where
\bea
A_3 & = & u_4\mu^3+u_3\mu^2+u_2\mu+u_1, \label{A3} \\
A_2 & = & (3u_4\mu^2+2u_3\mu+u_2)\beta+ u_7\mu^2+u_6\mu+u_5, \label{A2}\\
A_1 & = & (3u_4\mu+u_3)\beta^2+(2u_7\mu+u_6)\beta+u_9\mu+u_8, \label{A1} \\
A_0 & = & u_4\beta^3+u_7\beta^2+u_9\beta+u_{10}.
\eea
Moreover, $A_3=0$ since $F\in\bar\cE$. Thus, for $x_1$ we get a quadratic equation. By the Vieta formula, we have: $x_0+x_1=-\alpha$, where $\alpha=A_1/A_2$, which is the desired formula.
\end{proof}

\begin{lemma} \label{lemma involution finite}
Let $\bar\cE$ be a nonsingular cubic curve in $\bbC\bbP^2$ given in the non-homogeneous coordinates by the equation
$$
\cE: \quad u_1x^3+u_2x^2y+u_3xy^2+u_4y^3+u_5x^2+u_6xy+u_7y^2+u_8x+u_9y+u_{10}=0.
$$
Let $P_1=(x_1,y_1), P_2=(x_2,y_2)$ be two finite points on ${\mathcal{E}}$ with $x_1 \neq x_2$. If the point $P_3=I_{\bar\cE,P_1}(P_2)=I_{\bar\cE,P_2}(P_1)$ is finite, then it is given by
$$
P_3 = \big( -(x_1 +x_2 + \gamma), - (y_1 + y_2 + \delta)\big),
$$
where
\beq
\nu = \frac{y_2 -y_1}{x_2 -x_1}, \quad \beta= y_1 - \nu x_1=y_2-\nu x_2\nn,
\eeq
and
\bea
\gamma & = &  \frac{(3 u_4 \beta +u_7) \nu^2 + (2u_3 \beta +u_6)\nu +u_2 \beta +u_5}{u_4 \nu^3 +u_3 \nu^2 +u_2 \nu +u_1}, \label{gamma} \\
\delta & = &  \frac{u_7\nu^3-(u_3 \beta -u_6) \nu^2 - (2u_2 \beta -u_5)\nu -3u_1\beta}{u_4 \nu^3 +u_3 \nu^2 +u_2 \nu +u_1}. \label{delta}
\eea
\end{lemma}
\begin{proof}
Proceeding as before, we observe that in this case equation of the line $(P_1P_2)$ is $y=\nu x+\beta$. Thus, for the third intersection point $P_3=(x_3,y_3)$ of this line with $\cE$ we get a cubic equation (\ref{involution cubic eq}), where in the coefficients $A_i$ one should change notation $\mu$ to $\nu$.  We compute $x_3$ from the Vieta formula $x_1+x_2+x_3=-\gamma$, where $\gamma=A_2/A_3$ is found as in (\ref{gamma}). To compute $y_3$, we observe that
$$
y_3=\nu x_3+\beta=-\nu(x_1+x_2+\gamma)+\beta=-y_1-y_2+3\beta-\nu\gamma,
$$
and it remains to compute $\delta=\nu\gamma-3\beta$, which leads to (\ref{delta}).
\end{proof}

\section{Kahan discretization as Manin transformation}
\label{sect: 6-fold}

For computations described in this section, we will need concrete formulas for the map $\Phi_f$. The Hamiltonian vector field $f=J\nabla H$ for the Hamilton function
\beq \label{H}
H(x,y)= \tfrac{1}{3}a_1 x^3 +a_2 x^2y +a_3 x y^2 +\tfrac{1}{3}a_4 y^3+a_5 x^2 +2a_6 xy +a_7 y^2 +a_8 x+ a_9y
\eeq
is given by
\bea
\dot x & = & H_y\;=\; a_2x^2+2a_3xy+a_4y^2+2a_6x+2a_7y+a_9,  \label{k1} \\
\dot y & = & -H_x\;=\; -a_1x^2-2a_2xy-a_3y^2-2a_5x-2a_6y-a_8.    \label{k2}
\eea
(The normalization of the coefficients in the Hamilton function \eqref{H} is dictated by the desire to get rid of unnecessary factors in the equations of motion.)
The Kahan discretization of this system is the map $(\t x , \t y)= \Phi_f (x,y,\ep)$ defined by the equations of motion
\bea
(\t x-x)/\ep  & = & a_2 \t x x+a_3(\t xy+x \t y)+a_4\t yy+a_6(\t x+x)+a_7(\t y+y)+a_9,  \label{dk1}\\
(\t y-y)/\ep & = & -a_1\t xx-a_2(\t xy+x\t y)-a_3\t yy-a_5(\t x+x)-a_6(\t y +y)-a_8,        \label{dk2}
\eea
which can be solved for $\t x,\t y$ according to
\beq
\begin{pmatrix} \t x \\ \t y\end{pmatrix}=\begin{pmatrix} 1-\ep a_2 x-\ep a_3 y-\ep a_6 & -\ep a_3 x-\ep a_4 y-\ep a_7\\
\ep a_1 x+\ep a_2 y+\ep a_5  & 1+\ep a_2 x +\ep a_3 y+\ep a_6
\end{pmatrix}^{-1}\begin{pmatrix}  x+\ep a_6 x+\ep a_7 y +\ep a_9 \\
                                                      y-\ep a_5 x-\ep a_6 y-\ep a_8
\end{pmatrix}.
\eeq
As a result, 
\beq \label{eq: Kahan map}
\t x=\frac{R(x,y)}{D(x,y)}, \quad \t y=\frac{S(x,y)}{D(x,y)},
\eeq
where $R$, $S$ and $D$ are polynomials of degree 2. They can be found in Appendix \ref{Appendix}. 
The integral (\ref{eq: CMOQ integral}) in the case $n=2$ is given by 
\beq\label{eq: Kahan int}
\t H(x,y,\ep)= \frac{C(x,y,\ep)}{D(x,y,\ep)},
\eeq
where $C(x,y,\ep)$ is a polynomial in $x,y$ of degree 3, also given  in Appendix \ref{Appendix}. 

Consider the pencil of cubic curves $\frak E=\{\bar\cE_\lambda\}$ which are level sets of the integral $\t H$. Recall that this pencil has three base points $F_1, F_2, F_3$ at infinity, as well as six further base points $B_1,\ldots, B_6$ which are intersection points of the cubic curve $\bar C=0$ with the conic $\bar D=0$.

\begin{theorem}
To every base point $F$ at infinity, there correspond two base points $B$, $B'$, such that the Kahan map $\Phi$ is a Manin transformation in two different ways:
\beq \label{eq main}
\Phi = I_{\frak E,B} \circ I_{\frak E,F} = I_{\frak E,F} \circ I_{\frak E,B'}.
\eeq
\end{theorem}
Altogether, the Kahan map $\Phi$ is a Manin transformation in six different ways:
\bea
\Phi & = &  I_{\frak E,B_1} \circ I_{\frak E,F_1} = I_{\frak E,F_1} \circ I_{\frak E,B_4} \label{eq main 12}\\
       & = &  I_{\frak E,B_5} \circ I_{\frak E,F_2} = I_{\frak E,F_2} \circ I_{\frak E,B_2} \label{eq main 34}\\
       & = &  I_{\frak E,B_3} \circ I_{\frak E,F_3} = I_{\frak E,F_3} \circ I_{\frak E,B_6}. \label{eq main 56}
\eea

\begin{proof}
Both statements in \eqref{eq main} are proved similarly, therefore we restrict ourselves to the first one. We start with the following reformulation:
\begin{align}
& \Phi(P)=I_{\frak E,B} \circ I_{\frak E,F}(P)  \nonumber\\
& \qquad \Leftrightarrow \quad I_{\frak E,B}(\Phi(P))=I_{\frak E,F}(P) \nonumber\\
& \qquad \Leftrightarrow \quad I_{\frak E,\Phi(P)}(B)=I_{\frak E,P}(F) \nonumber\\
& \qquad \Leftrightarrow \quad B=I_{\frak E,\Phi(P)}\circ I_{\frak E,P}(F). \label{eq to show gen}
\end{align}
Thus, it is sufficient to prove that there exists a base point $B$ such that \eqref{eq to show gen} is satisfied for all $P$ (for which its right hand side is defined). Since the expression $I_{\frak E,\Phi(P)}\circ I_{\frak E,P}(F)$ is a continuous function of $P$ on its definition domain, which is a connected set, and since the set of base points is finite, it is sufficient to prove that, for any point $P$, its image $I_{\frak E,\Phi(P)}\circ I_{\frak E,P}(F)$ is a base point. Since translations of the argument by vectors in $\bbC^2$ act transitively on the set of cubic Hamilton functions, it is sufficient to prove the statement just for one point $P_0$ (and for all Hamilton functions). We will choose $P_0=(0,0)\in\bbC^2$. Thus, the proof of the theorem boils down to the following statement.
\begin{lemma}
The point $I_{\frak E,\Phi(0,0)}\circ I_{\frak E,(0,0)}(F)$ is a base point of the pencil $\frak E$. More precisely, the polynomial $D(x,y,\ep)$ vanishes at this point.
\end{lemma}
This lemma is proved by a direct computation which is outlined as follows.
\smallskip

1) Determine the pencil parameter of the level curve containing the point $(0,0)$, that is, $\lambda_0=\t H(0,0,\epsilon)=c_{10}/d_6$. 

2) Compute the coefficients of the level curve $\bar\cE_{\lambda_0}$, which are $u_i= c_i$ for $i=1,\ldots,4$, and $u_i= c_i-\lambda_0d_{i-4}$ for $i=5,\ldots,10$. Of course, $u_{10}=0$. 

3) Compute, according to Lemma \ref{lemma involution inf}:
\beq
P_1=I_{\bar\cE_{\lambda_0},(0,0)}(F)=(x_1,y_1),
\eeq
where
\beq\label{P1}
x_1= -\frac{u_9\mu+u_8}{u_7 \mu^2 + u_6 \mu +u_5}, \quad y_1= -\frac{u_9\mu^2+u_8\mu}{u_7 \mu^2 + u_6 \mu +u_5}.  
\eeq

4) Compute 
\beq
P_2=\Phi(0,0)=(x_2,y_2),
\eeq
where
\beq\label{P2}
x_2=\frac{\ep a_9+\ep^2(a_6a_9-a_7a_8)}{1+\ep^2( a_5 a_7-a_6^2 )}, \quad
y_2=\frac{-\ep a_8+\ep^2(a_6a_8-a_5a_9)}{1+\ep^2( a_5 a_7-a_6^2 )}.
\eeq

5) Compute $I_{\frak E,\Phi(0,0)}\circ I_{\frak E,(0,0)}(F)$ by applying Lemma \ref{lemma involution finite} with the point $P_1=(x_1,y_1)$ given in (\ref{P1}) and with the point $P_2=(x_2,y_2)$ given in (\ref{P2}). The coordinates of the resulting point $P_3=I_{\bar\cE_{\lambda_0},P_1}(P_2)$ are rational functions of $\mu$ with numerators and denominators of degree 8. 

6) Compute the value of the quadratic polynomial $D(x,y,\ep)$ at the point $P_3$. This value is a rational function of $\mu$. Its numerator is a polynomial of $\mu$ of degree 16. Maple computation shows that this polynomial is divisible by $u_4\mu^3+u_3\mu^2+u_2\mu+u_1$.
The following tricks can be used to lower the complexity of computations. First, for homogeneity reasons, it is sufficient to take $\ep=1$. Second, we can use affine transformations of the plane $(x,y)$ (under which the Kahan discretization is covariant) to achieve vanishing of two of the coefficients of the polynomial $H$, for instance, $a_7=0$ and $a_9=0$.  

As a consequence of this remarkable factorization, the numerator of $D(P_3,\ep)$ vanishes, as soon as $F$ is a base point at infinity, that is, as soon as $u_4\mu^3+u_3\mu^2+u_2\mu+u_1=0$. 
\end{proof}

\section{Geometry of the pencil of elliptic curves for the Kahan discretization}
\label{sect: geometry}

We now study the geometric consequences of the six-fold representation of the Kahan discretization as Manin maps. We restrict ourselves to the generic case, when the base points $B_1,\ldots,B_6$ are finite and pairwise distinct.
\begin{theorem}\label{Th geometry}
The lines $(B_1B_2)$ and $(B_4B_5)$ are parallel and pass through the point $F_3$ at infinity. Similarly, the lines $(B_2B_3)$ and $(B_5B_6)$ are parallel and pass through the point $F_1$ at infinity, and the lines $(B_3B_4)$ and $(B_6B_1)$ are parallel and pass through the point $F_2$ at infinity.
\end{theorem}
\begin{proof}
Consider an arbitrary invariant curve of the pencil, along with the addition law on this curve (assuming that the neutral element $O$ is a flex, so that $P_1+P_2+P_3=O$ is equivalent to collinearity of the points $P_1$, $P_2$, $P_3$). In particular, since all three points $F_1$, $F_2$, $F_3$ lie on the line at infinity, we have:
\beq
F_1+F_2+F_3=O.
\eeq
Now write down equations \eqref{eq main 12}-\eqref{eq main 56} in terms of this addition law. We have:
\beq
F_1-B_1= B_2-F_2=F_3-B_3=B_4-F_1=F_2-B_5=B_6-F_3.
\eeq
Consider, for instance, the first of these equations. We have:
$$
B_1+B_2=F_1+F_2=-F_3 \quad \Leftrightarrow\quad B_1+B_2+F_3=O.
$$
Thus, the straight line $(B_1B_2)$ passes through $F_3$. See Figure \ref{Fig2}.
\end{proof}

\begin{figure}
\begin{center}
\includegraphics[width=0.5\textwidth]{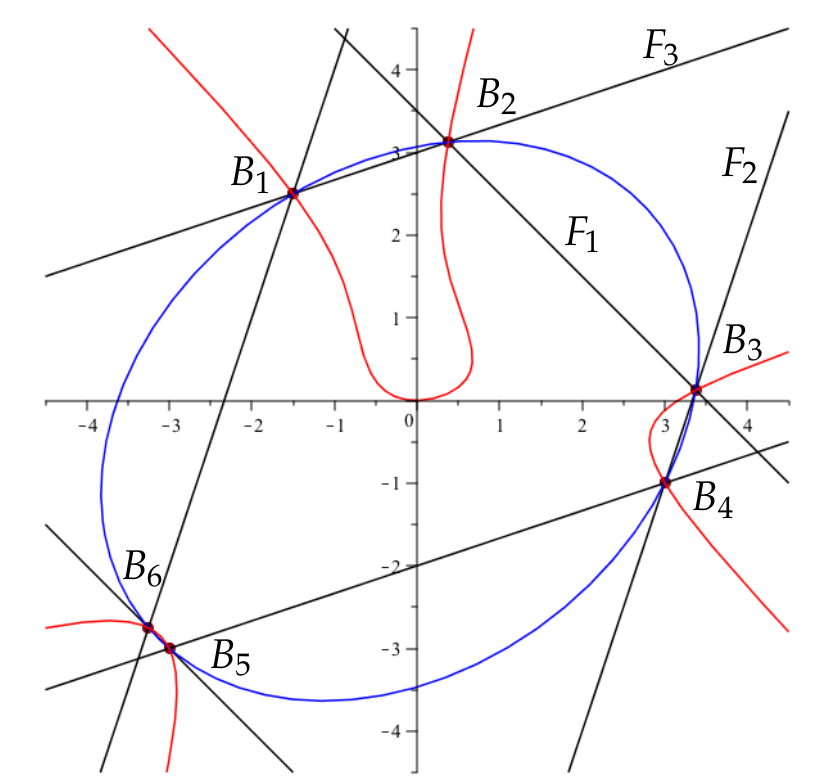}
\caption{The opposite side lines of the hexagon $B_1B_2B_3B_4B_5B_6$ are parallel.}
\label{Fig2}
\end{center}
\end{figure}

\begin{corollary}
For each cubic curve of the pencil, the second intersection points of the curve with the tangents at the points $F_1$, $F_2$, $F_3$ at infinity, lie on the lines connecting the corresponding base points, $(B_1B_4)$, $(B_2B_5)$, and $(B_3B_6)$, respectively.
\end{corollary}
\begin{proof}
This follows from the equations like 
$$
B_2-F_2=F_2-B_5\quad \Leftrightarrow\quad  B_2+B_5+(-2F_2)=O,
$$
which mean that the points $B_2$, $B_5$ and $-2F_2$ are collinear. Recall that $-2F_2$ is the second intersection point of the cubic curve with its tangent at $F_2$.  
See Figure \ref{Fig3}.
\end{proof}

\begin{figure}
\begin{center}
\includegraphics[width=0.5\textwidth]{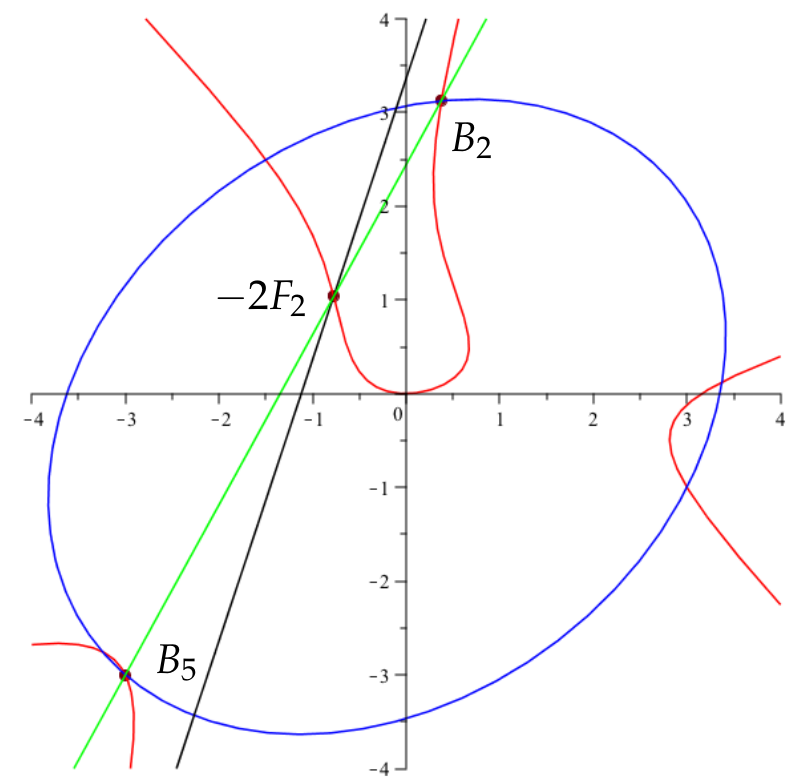}
\caption{The (green) line $(B_2B_5)$ passes through the point $-2F_2$ which is the second intersection point of the (black) tangent line to the curve $\bar C=0$ at the point $F_2$ at infinity}
\label{Fig3}
\end{center}
\end{figure}

A very remarkable statement is an inverse theorem to Theorem \ref{Th geometry}. 
\begin{theorem}\label{Th characterization}
Let a hexagon $B_1B_2B_3B_4B_5B_6$ have three pairs of parallel side lines $(B_1B_2)\parallel (B_4B_5)$, $(B_2B_3)\parallel (B_5B_6)$, and $(B_3B_4)\parallel (B_6B_1)$, passing through the points $F_3$, $F_1$, and $F_2$ at infinity, respectively. Consider the pencil $\frak E$ of cubic curves with the base points $\{B_1,\ldots,B_6,F_1,F_2,F_3\}$. Then the birational map $\Phi$ defined by the pencil $\frak E$ according to equations \eqref{eq main 12}--\eqref{eq main 56} is the Kahan discretization of a planar Hamiltonian quadratic vector field.
\end{theorem}
\begin{proof}
Observe that, according to the Pascal's theorem, the points $B_1,\ldots,B_6$ lie on a conic $D$. The pencil $\frak E$ contains a reducible curve consisting of the conic $D$ and the line at infinity. The fact that all six Manin maps in \eqref{eq main 12}--\eqref{eq main 56} coincide is an easy consequence of the condition of the theorem, if one reverses the arguments from the proof of Theorem \ref{Th geometry}. It remains to prove the statement for one of the representations, say for $I_{\frak E, B_1}\circ I_{\frak E,F_1}$. This is again done with a direct Maple computation, which can be organized as follows. 
\smallskip

1) Prescribe arbitrary nine coefficients  of the side lines of the hexagon (three slopes $\mu_1$, $\mu_2$, $\mu_3$ and six heights $b_1,\ldots,b_6$), so that equations of these lines read:
$$
(B_1B_2): \; y=\mu_3x+b_1, \quad (B_2B_3): \; y=\mu_1x+b_2, \quad (B_3B_4): \; y=\mu_2x+b_3, 
$$
$$
(B_4B_5): \; y=\mu_3x+b_4, \quad (B_5B_6): \; y=\mu_1x+b_5, \quad (B_6B_1): \; y=\mu_2x+b_6. 
$$

2) Compute the points $B_1,\ldots,B_6$ as pairwise intersection points of these lines:
$$
B_1= \left(\frac{b_1-b_6}{\mu_2-\mu_3},\frac{\mu_2b_1-\mu_3b_6}{\mu_2-\mu_3}\right), \quad {\rm etc.}
$$

3) Compute the coefficients $d_1,\ldots,d_6$ of the conic passing through these six points:
$$
D(x,y)=d_1x^2+d_2xy+d_3y^2+d_4x+d_5y+d_6=0.
$$

4) Compute the coefficients $c_5,\ldots,c_{10}$ of the pencil of cubic curves $C(x,y)-\lambda D(x,y)=0$ passing through $B_1,\ldots, B_6$.  Here
$$
 C(x,y)=(y-\mu_1x)(y-\mu_2x)(y-\mu_3x)+c_5x^2+c_6xy+c_7y^2+c_8x+c_9y+c_{10}=0
$$ 
is the equation of an arbitrarily chosen nonsingular curve of the pencil. For instance, setting $c_{10}=0$ defines $c_5,\ldots,c_9$ uniquely.

5) Compute the involutions $I_{\bar\cE,F_1}(x,y)$ and $I_{\bar \cE,B_1}(x,y)$ according to formulas from Lemmas
\ref{lemma involution inf} and \ref{lemma involution finite} for the curve $\cE: C-\lambda D=0$ of the pencil passing through a running point $(x,y)$. Thus, we use the expressions
$$
u_1=-\mu_1\mu_2\mu_3,\quad u_2=\mu_1\mu_2+\mu_2\mu_3+\mu_3\mu_1, \quad u_3=-(\mu_1+\mu_2+\mu_3), \quad u_4=1,
$$
and
$$
u_i=c_i-\lambda d_{i-4}, \quad i=5,\ldots,10, \quad {\rm where}\quad \lambda=\frac{C(x,y)}{D(x,y)}.
$$

6) Compute the map $(\t x, \t y)=\Phi(x,y)=I_{\bar \cE,B_1}\circ I_{\bar\cE,F_1}(x,y)$. It turns out to be of the form
$$
\t x=\frac{R(x,y)}{D(x,y)}, \quad \t y=\frac{S(x,y)}{D(x,y)},
$$
where $R$ and $S$ are polynomials of degree 2.
\smallskip

7) Check whether these rational functions $\t x$, $\t y$ satisfy equations of motion of the Kahan discretization with $\ep=1$:
\bea
\t x-x  & = & a_2 \t x x+a_3(\t xy+x \t y)+a_4\t yy+a_6(\t x+x)+a_7(\t y+y)+a_9, \nonumber \\
\t y-y & = & -a_1\t xx-a_2(\t xy+x\t y)-a_3\t yy-a_5(\t x+x)-a_6(\t y +y)-a_8, \nonumber
\eea
with some $a_1,\ldots, a_9$. This is a system of linear equations for the coefficients $a_1,\ldots, a_9$, which turns out to have a unique solution.
This solution can be found in Appendix \ref{Appendix B}. It leads to the following compact formula for the Hamilton function $H(x,y)$ in terms of the data $\mu_i$, $b_i$ of the pencil of cubic curves (which stay invariant under the map $\Phi_f(x,y;\ep)$ with $f={J\nabla H}$ and $\ep=1$):
\bea
H(x,y) & = & \frac{2\mu_{12}}{b_{14}\mu_{23}\mu_{13}}\Big(\tfrac{1}{3}(\mu_3x-y)^3+\tfrac{1}{2}(b_1+b_4)(\mu_3x-y)^2+b_1b_4(\mu_3x-y)\Big) \nn \\
  &  & -\frac{2\mu_{23}}{b_{25}\mu_{12}\mu_{13}}\Big(\tfrac{1}{3}(\mu_1x-y)^3+\tfrac{1}{2}(b_2+b_5)(\mu_1x-y)^2+b_2b_5(\mu_1x-y)\Big) \nn \\
  & & +\frac{2\mu_{13}}{b_{36}\mu_{12}\mu_{23}}\Big(\tfrac{1}{3}(\mu_2x-y)^3+\tfrac{1}{2}(b_3+b_6)(\mu_2x-y)^2+b_3b_6(\mu_2x-y)\Big), \nn
\eea
where $b_{ij}=b_i-b_j$, $\mu_{ij}=\mu_i - \mu_j$.
\end{proof}

\section{Conclusions}

We proved an amazing characterization of integrable maps arising as Kahan discretizations of quadratic planar Hamiltonian vector fields, in terms of the geometry of their set of invariant curves. Such a neat characterization supports our belief expressed in \cite{PPS1, PPS2} that Kahan-Hirota-Kimura discretizations will serve as a rich source of novel results concerning algebraic geometry of integrable birational maps. It will be desirable to find similar characterizations for further classes of integrable Kahan discretizations, in dimensions $n=2$ (which belongs to our agenda in the near future) and for $n>2$ (even more important but also more difficult) .

\begin{appendix}
\section{Formulas for the planar Kahan map and its integral}\label{Appendix}

The numerators of the components of map \eqref{eq: Kahan map} are:
\bea 
R(x,y,\ep) & = & r_1 x^2 + r_2 xy + r_3 y^2 + r_4 x + r_5 y + r_6, \label{eq:R}\\
S(x,y,\ep) & = & s_1 x^2 + s_2 xy + s_3 y^2 + s_4 x + s_5y + s_6, \label{eq:S}
\eea
with
\bea 
&& r_1 = \ep a_2 + \ep^2(a_2 a_6 - a_3 a_5), \nn\\
&& r_2 = 2\ep a_3 + \ep^2(a_2 a_7 - a_4 a_5), \nn \\
&& r_3 = \ep a_4 + \ep^2(a_3 a_7 - a_4 a_6 ), \nn \\
&& r_4 = 1 + 2\ep a_6 + \ep^2(a_6^2 - a_5 a_7 - a_3 a_8 + a_2 a_9), \nn \\
&& r_5 = 2\ep a_7 + \ep^2(a_3 a_9 - a_4 a_8), \nn\\
&& r_6 = \ep a_9 + \ep^2(a_6 a_9 - a_7 a_8),\nn 
\eea 
and
\bea 
&& s_1 = -\ep a_1 +\ep^2(a_2 a_5- a_1 a_6), \nn\\ 
&& s_2 = -2\ep a_2 + \ep^2(a_3 a_5 - a_1 a_7), \nn \\
&& s_3 = -\ep a_3 + \ep^2(a_3 a_6 - a_2 a_7). \nn \\
&& s_4 = -2\ep a_5 + \ep^2(a_2 a_8 - a_1 a_9), \nn\\
&& s_5 = 1 - 2\ep a_6 + \ep^2(a_6^2 - a_5 a_7 + a_3 a_8 - a_2 a_9), \nn \\
&& s_6 = -\ep a_8 + \ep^2(a_6 a_8 - a_5 a_9) .\nn 
\eea
The common denominator is:
\beq \label{H denom}
D(x,y,\ep)=  d_1 x^2 +d_2 xy + d_3 y^2 +d_4 x+ d_5y + d_6, 
\eeq
where
\bea
&& d_1 = \ep^2(a_1 a_3 -a_2^2 ), \nn \\
&& d_2 = \ep^2(a_1 a_4-a_2 a_3),\nn \\
&& d_3 = \ep^2(a_2 a_4- a_3^2),\nn \\
&& d_4 = \ep^2(a_1 a_7 - 2a_2 a_6+ a_3 a_5),\nn \\
&& d_5 = \ep^2(a_4 a_5 - 2a_3 a_6+ a_2 a_7),\nn \\
&& d_6 = 1+\ep^2(a_5 a_7-a_6^2 ). \nn
\eea
Similarly,
\beq  \label{H num}
C(x,y,\ep)=  c_1 x^3 +c_2 x^2y +c_3 x y^2 +c_4 y^3
+c_5 x^2 +c_6 xy +c_7 y^2 +c_8 x+ c_9y + c_{10},
\eeq
with the coefficients
\bea
&& \t c_1=  \tfrac{1}{3}a_1+\tfrac{1}{3}\ep^2(a_1 a_3 a_8 -a_1 a_6^2 +2 a_2 a_5 a_6-a_2^2a_8 -a_3a_5^2), \nn \\
&& \t c_2 = a_2+\tfrac{1}{3}\ep^2(a_1 a_3 a_9 + a_1 a_4 a_8 - 2 a_1 a_6 a_7 + 2 a_2 a_5 a_7 - a_4 a_5^2 - a_2^2 a_9 - a_2 a_3 a_8 + a_2 a_6^2) , \nn \\
&& \t c_3 = a_3+\tfrac{1}{3}\ep^2(a_2 a_4 a_8 + a_1 a_4 a_9 - 2 a_4 a_5 a_6 + 2 a_3 a_5 a_7 - a_1 a_7^2 - a_3^2 a_8 - a_2 a_3 a_9 + a_3 a_6^2) , \nn \\
&& \t c_4 = \tfrac{1}{3}a_4+\tfrac{1}{3}\ep^2(a_2 a_4 a_9  - a_4 a_6^2 + 2 a_3 a_6 a_7 - a_3^2a_9 -  a_2 a_7^2) , \nn \\
&& \t c_5 = a_5+\tfrac{1}{3}\ep^2(a_1 a_7 a_8 - 2  a_1 a_6 a_9 + 2 a_2 a_5 a_9 - a_3 a_5 a_8 - a_5^2 a_7 + a_5 a_6^2) , \nn \\
&& \t c_6 = 2a_6+\tfrac{1}{3}\ep^2(-a_1 a_7 a_9 -a_4 a_5 a_8 - 2 a_2 a_6 a_9 - 2 a_3 a_6 a_8 + 3 a_2 a_7 a_8 +3 a_3 a_5 a_9 -  a_5 a_6 a_7 + 2a_6^3), \nn \\
&&\t c_7 =  a_7+\tfrac{1}{3}\ep^2(a_4 a_5 a_9 - 2 a_4 a_6 a_8 + 2 a_3 a_7 a_8 -  a_2 a_7 a_9 - a_5 a_7^2 + a_6^2 a_7), \nn \\
&& \t c_8 = a_8+\tfrac{1}{3}\ep^2(2 a_2 a_8 a_9 - a_1 a_9^2 - a_3 a_8^2 - a_5 a_7 a_8 + a_6^2 a_8),  \nn \\
&& \t c_9 = a_9+\tfrac{1}{3}\ep^2(2 a_3 a_8 a_9 - a_4 a_8^2 - a_2 a_9^2 - a_5 a_7 a_9 + a_6^2 a_9), \nn \\
&& \t c_{10} = \tfrac{1}{3}\ep^2(2 a_6 a_8 a_9 - a_5 a_9^2  - a_7 a_8^2).\nn 
\eea

\section{Coefficients of the Hamiltonian from the data of the pencil}\label{Appendix B}

Here are the formulas referred to at the end of the proof of Theorem \ref{Th characterization}:
\bea
a_1&=& \frac{
b_{25} b_{36}\mu_{12}^2 \mu_3^3
-b_{14}b_{36}\mu_{23}^2 \mu_1^3
+ b_{14} b_{25}\mu_{13}^2 \mu_2^3
}{D}, \nn \\
a_2&=& \frac{
-  b_{25} b_{36}\mu_{12}^2 \mu_3^2
+b_{14}b_{36}\mu_{23}^2 \mu_1^2
- b_{14} b_{25}\mu_{13}^2 \mu_2^2
}{D}, \nn \\
a_3&=& \frac{
b_{25} b_{36}\mu_{12}^2 \mu_3
-b_{14}b_{36}\mu_{23}^2 \mu_1
+ b_{14} b_{25}\mu_{13}^2 \mu_2
}{D}, \nn \\
a_4&=& \frac{
 - b_{25} b_{36}\mu_{12}^2
+b_{14}b_{36}\mu_{23}^2 
- b_{14} b_{25}\mu_{13}^2 
}{D}, \nn \\
a_5&=& \frac{
(b_1+b_4) b_{25} b_{36} \mu_3^2 \mu_{12}^2
- (b_2+b_5) b_{14} b_{36} \mu_1^2 \mu_{23}^2
+ (b_3+b_6) b_{14} b_{25} \mu_2^2 \mu_{13}^2
}{2D}, \nn \\
a_6&=& \frac{
- (b_1+b_4) b_{25} b_{36} \mu_3 \mu_{12}^2
+ (b_2+b_5) b_{14} b_{36} \mu_1 \mu_{23}^2
- (b_3+b_6) b_{14} b_{25} \mu_2 \mu_{13}^2
}{2D}, \nn \\
a_7&=& \frac{
(b_1+b_4) b_{25} b_{36}  \mu_{12}^2
- (b_2+b_5) b_{14} b_{36}  \mu_{23}^2
+(b_3+b_6) b_{14} b_{25}  \mu_{13}^2 }{2D}, \nn \\
a_8&=& \frac{ 
b_1 b_4 b_{25} b_{36} \mu_3 \mu_{12}^2
- b_2 b_5 b_{14} b_{36} \mu_1 \mu_{23}^2
+ b_3 b_6 b_{14} b_{25}  \mu_2 \mu_{13}^2
}{D}, \nn \\
a_9&=& \frac{
- b_1 b_4 b_{25} b_{36}  \mu_{12}^2
+ b_2 b_5 b_{14} b_{36}  \mu_{23}^2
- b_3 b_6 b_{14} b_{25}  \mu_{13}^2
}{D}, \nn 
\eea
where
$$
D=\frac12 b_{14}b_{25}b_{36}\mu_{12}\mu_{13}\mu_{23}
$$
and $b_{ij}=b_i-b_j$, $\mu_{ij}=\mu_i - \mu_j$. 

\end{appendix}

\section*{Acknowledgment}
This research is supported by the DFG Collaborative Research Center TRR 109 ``Discretization in Geometry and Dynamics''.


\end{document}